\documentclass[a4paper,fleqn]{cas-sc}

\usepackage[numbers,sort&compress]{natbib}
\usepackage{amsmath,amssymb,amsthm,mathtools,bm}
\usepackage{graphicx}
\usepackage{booktabs,tabularx,array}
\usepackage{enumitem}
\usepackage{microtype}
\usepackage{tikz}
\usetikzlibrary{arrows.meta,positioning}

\newtheorem{theorem}{Theorem}
\newtheorem{proposition}[theorem]{Proposition}

\newtheorem{lemma}[theorem]{Lemma}
\theoremstyle{definition}
\newtheorem{definition}[theorem]{Definition}
\newtheorem{example}[theorem]{Example}
\theoremstyle{remark}
\newtheorem{remark}[theorem]{Remark}

\newcommand{\C}{\mathbb C}
\newcommand{\one}{\bm 1}
\newcommand{\kdiag}{\kappa_{\mathrm{diag}}}
\newcommand{\tdiag}{\tau_{\mathrm{diag}}}

\begin{document}
\let\WriteBookmarks\relax
\def\floatpagepagefraction{1}
\def\textpagefraction{.001}

\shorttitle{Reduced arc-weight heterogeneity for directed-network master stability}
\shortauthors{A. B. S. Arokiadoss and G. Arunkumar}

\title[mode=title]{Reduced Arc-Weight Heterogeneity for Master Stability Analysis of Directed Networks}

\author[1]{Aandrew Baggio Sahaya Arokiadoss}
\cormark[1]
\ead{ee20d067@iitm.ac.in}
\credit{Conceptualization, Methodology, Formal analysis, Investigation, Writing -- original draft, Writing -- review and editing}

\affiliation[1]{organization={Department of Electrical Engineering, Indian Institute of Technology Madras},
  city={Chennai},
  postcode={600036},
  state={Tamil Nadu},
  country={India}}

\author[2]{G. Arunkumar}
\ead{garunkumar@iitm.ac.in}
\credit{Supervision, Writing -- review and editing}

\affiliation[2]{organization={Department of Mathematics, Indian Institute of Technology Madras},
  city={Chennai},
  postcode={600036},
  state={Tamil Nadu},
  country={India}}

\cortext[1]{Corresponding author}

\begin{abstract}
Master stability analysis can be applied to directed networks even when their Laplacian matrices are not diagonalizable. Nevertheless, a diagonalizable Laplacian is easier to analyse because the perturbations governing synchronization separate into independent modes. This paper asks how many different positive coupling strengths are needed to obtain this simplification without adding or removing any directed connections. We first give an elementary construction showing that suitable positive weights exist for every weakly connected directed network and explain its relationship to stronger genericity results in the literature. We then introduce a measure of the smallest number of distinct coupling strengths required for diagonalizability. If a common coupling strength already produces a diagonalizable Laplacian, only one value is needed. For directed acyclic networks, we give a constructive upper bound based on the length of the longest directed path from a source. This bound is exact for rooted outward directed trees. Consequently, an outward star requires only one coupling value, whereas a directed path requires a different value on every arc. These results reduce coupling heterogeneity while retaining the simpler modal form of master stability analysis. Diagonalizability is treated as an analytical convenience, not as a sufficient condition for synchronization.
\end{abstract}

\begin{highlights}
\item Diagonalizable Laplacians yield independent master-stability modes.
\item We minimize the weight values needed for Laplacian diagonalizability.
\item Longest-path levels give an upper bound for acyclic digraphs.
\item The bound is exact for outward directed trees.
\item Stars need one value, while directed paths need one value per arc.
\end{highlights}


\maketitle

\section{Introduction}

Synchronization in networked dynamical systems depends jointly on the node dynamics, the coupling law, and the interaction network \citep{Strogatz2001,Boccaletti2006,Arenas2008}.  The master stability function (MSF) of Pecora and Carroll separates the oscillator-dependent variational equation from the network-dependent Laplacian spectrum \citep{PecoraCarroll1998}.  In the diagonalizable case, a change to the Laplacian eigenbasis yields independent perturbation modes.  Directed Laplacians need not be diagonalizable, but this does not prevent master-stability analysis: Nishikawa and Motter extended the formalism to nondiagonalizable network structures using Jordan form and showed, in particular, that high synchronizability need not coincide with diagonalizability \citep{NishikawaMotter2006}.  The role of diagonalizability in this paper is therefore one of analytical convenience---it restores the completely decoupled modal representation---rather than a claimed extension of the MSF to directed networks.

Positive edge and arc weights provide a natural network-design freedom.  They have been optimized for synchronization in directed networks \citep{SkardalTaylorSun2016} and used more broadly to shape Laplacian spectra and network dynamics \citep{Forrow2018}.  Closest to the present spectral question is the structural-genericity theory of Poignard, Pereira, and Pade, which shows that favorable spectral properties, including simple spectra under the directed rooted-spanning-tree hypothesis developed in that work, are generic with respect to perturbations of existing positive weights \citep{Poignard2018}.  Their framework is stronger than a bare existence statement in the one-root case.  Related work subsequently connected such structural perturbations with changes in synchronizability \citep{PadePoignardPereira2019}.

Those results motivate a different design question.  If nonuniform positive weights are used to obtain a diagonalizable directed Laplacian, allowing every arc to take an independent value may introduce much more heterogeneity than necessary.  Several arcs may safely reuse one coupling value.  We therefore ask
\begin{quote}
What is the smallest number of \emph{distinct positive arc-weight values} needed to make the in-degree Laplacian diagonalizable?
\end{quote}
This is not the same as counting how many arcs must be perturbed away from a common value: several such arcs may ultimately share one value.

We formalize the question through a graph parameter $\kdiag(G)$.  Before doing so, we record a simple topology-preserving existence construction for weakly connected digraphs.  It uses a spanning DAG, distinct positive weighted indegrees on the DAG, and a sufficiently small positive restoration of omitted arcs.  The zero eigenvalue remains semisimple because its algebraic and geometric multiplicities agree for nonnegatively weighted directed Laplacians \citep{AgaevChebotarev2005,CaughmanVeerman2006}.  This construction is included to make the domain of the heterogeneity problem explicit; it is not intended to supersede the stronger genericity result of Poignard \emph{et al.} in the one-root case \citep{Poignard2018}.

The new emphasis is on how few weight values suffice.  First, if the uniformly weighted Laplacian is already diagonalizable, then $\kdiag(G)=1$.  For DAGs, vertices are grouped by longest-path level and all arcs entering one level are assigned a common value.  A geometric separation of the level weights distinguishes eigenvalues belonging to different levels.  Repetitions within a level are allowed and are shown to be semisimple because no directed path joins two same-level vertices.  This yields
\[
\kdiag(G)\le h(G)
\]
for every DAG.  For rooted outward directed trees the bound is exact: $\kdiag(T)=h(T)$.  Consequently, the outward star needs one weight value while the directed path on $n$ vertices needs $n-1$.

The resulting problem is naturally viewed as a weight-class minimization problem: partition the existing arc set into as few classes as possible, assign one positive value to each class, and require the resulting in-degree Laplacian to be diagonalizable.  In an MSF setting this parameter measures the amount of coupling heterogeneity introduced solely to recover the simpler diagonalizable modal representation; the actual synchronization condition still depends on the MSF stability region and on the designed eigenvalues.

\section{Directed-network master stability and weighted Laplacians}
\label{sec:setup}

Let $G=(V,E)$ be a finite loopless digraph with $n$ vertices. To each arc $u\to v$ assign a positive weight $w_{uv}>0$. The weighted indegree of $v$ is
\begin{equation}
 d_w^-(v)=\sum_{u:\,u\to v} w_{uv}.
\end{equation}
The unweighted indegree is
\begin{equation}
 d^-(v)=|\{u:u\to v\}|.
\end{equation}

Using the indegree convention, define the weighted adjacency matrix $Q$ by
\begin{equation}
 Q_{vu}=\begin{cases}
 w_{uv},&u\to v\in E,\\
 0,&\text{otherwise},
 \end{cases}
\end{equation}
and let
\begin{equation}
 D=\operatorname{diag}\bigl(d_w^-(v_1),\ldots,d_w^-(v_n)\bigr).
\end{equation}
The weighted indegree Laplacian is
\begin{equation}
 L(w)=D-Q.
\label{eq:laplacian}
\end{equation}
Thus $L(w)\one=0$. Diagonalizability is understood over $\C$.

For context, consider diffusively coupled identical systems
\begin{equation}
 \dot x_i=f(x_i)-\sigma\sum_{j=1}^{n}L_{ij}(w)H(x_j),
 \qquad i=1,\ldots,n,
\label{eq:network}
\end{equation}
with synchronous solution $x_i(t)=s(t)$, $\dot s=f(s)$. Linearization about the synchronization manifold gives
\begin{equation}
 \dot\xi=\bigl[I_n\otimes Df(s)-\sigma L(w)\otimes DH(s)\bigr]\xi.
\label{eq:variational}
\end{equation}
If $L(w)=S\Lambda S^{-1}$ is diagonalizable, the transformation $\eta=(S^{-1}\otimes I)\xi$ gives independent modal equations
\begin{equation}
 \dot\eta_i=\bigl[Df(s)-\sigma\lambda_iDH(s)\bigr]\eta_i.
\label{eq:modes}
\end{equation}
This is the familiar decoupled form underlying the standard MSF construction \citep{PecoraCarroll1998}.  If $L(w)$ is defective, the corresponding Jordan-form extension produces coupled generalized modes rather than independent equations \citep{NishikawaMotter2006}.

A repeated eigenvalue is not an obstruction by itself. A matrix is diagonalizable exactly when every eigenvalue has equal algebraic and geometric multiplicities. Repeated eigenvalues are therefore allowed provided that they are semisimple.  In particular, for nonnegative directed Laplacians the zero eigenvalue is semisimple; its multiplicity is determined by the directed reach/root structure \citep{AgaevChebotarev2005,CaughmanVeerman2006}.

\section{A positive diagonalizability baseline for weakly connected digraphs}
\label{sec:existence}

The heterogeneity parameter introduced below is finite for every weakly connected digraph.  The following elementary construction supplies a convenient witness.

\begin{lemma}[Spanning DAG]
\label{lem:spanningdag}
If $G$ is weakly connected and has $r$ root strongly connected components, then it contains a weakly connected spanning DAG with exactly $r$ sources, one in each root SCC.
\end{lemma}

\begin{proof}
Contract the SCCs to obtain the condensation DAG and retain one representative original arc for every selected condensation connection needed to span the condensation.  In each root SCC choose one vertex and a spanning out-arborescence.  In each non-root SCC retain an incoming representative arc, use its head as a local root, and retain a spanning out-arborescence from that vertex.  The union spans all vertices and is weakly connected.  It is acyclic because cycles are excluded inside each retained arborescence and a cycle between SCCs would induce a cycle in the condensation.  Its only sources are the selected vertices in the root SCCs.
\end{proof}

\begin{proposition}[Weighting the spanning DAG]
\label{prop:spanningdagweights}
Let $D_G$ be such a spanning DAG with $r$ sources.  Assign pairwise distinct positive target indegrees to its $n-r$ non-source vertices.  The arcs entering each non-source vertex can be given positive weights whose sum equals its target indegree.  In a topological ordering, the in-degree Laplacian of $D_G$ is triangular with those target indegrees as its nonzero diagonal entries.  Hence its nonzero eigenvalues are positive and simple, while zero is semisimple of multiplicity $r$.
\end{proposition}

\begin{proof}
Split the target indegree of each non-source vertex into positive summands over the arcs entering that vertex.  In topological order the Laplacian is lower triangular, so its diagonal entries are zero at the sources and are the prescribed weighted indegrees elsewhere.  The nonzero entries are simple by construction.  The $r$ zero eigenvalues are semisimple for the directed Laplacian \citep{AgaevChebotarev2005,CaughmanVeerman2006}.
\end{proof}

Give every omitted arc a common weight $\varepsilon>0$ while retaining the DAG weights.  The resulting Laplacian depends continuously on $\varepsilon$.  The simple nonzero eigenvalues of the DAG therefore remain simple for all sufficiently small $\varepsilon>0$, while the root-SCC structure fixes the semisimple zero multiplicity.  We obtain the following baseline statement.

\begin{theorem}[Positive diagonalizing weights]
\label{thm:existence}
Every weakly connected digraph admits a strictly positive arc-weight assignment for which its weighted in-degree Laplacian is diagonalizable over $\C$.
\end{theorem}

\begin{remark}[Relation to structural genericity]
When the digraph has a single root SCC, the existence conclusion above sits inside the substantially stronger structural-genericity theory of Poignard, Pereira, and Pade, which shows that simple spectrum is generic under topology-preserving perturbations of the existing weights under their directed rooted-spanning-tree hypothesis \citep{Poignard2018}.  The elementary construction is retained here only as a concrete starting point for the subsequent weight-class question and to state the diagonalizability baseline when repeated semisimple zero eigenvalues are unavoidable.
\end{remark}

\section{Minimum diagonalizing arc-weight heterogeneity}
\label{sec:heterogeneity}

\begin{definition}[Minimum diagonalizing weight heterogeneity]
Let $G=(V,E)$ be a digraph. If $E=\varnothing$, define $\kdiag(G)=0$. If $E\ne\varnothing$ and at least one strictly positive arc-weight assignment makes $L(w)$ diagonalizable, define
\begin{equation}
\kdiag(G)=
\min_{w>0\, :\, L(w)\ \mathrm{diagonalizable}}
\left|\{w_e:e\in E\}\right|.
\label{eq:kdiag}
\end{equation}
\end{definition}

For a nonempty arc set, the value $\kdiag(G)$ is the minimum number of distinct positive numbers that must appear among the arc weights. Thus $\kdiag(G)=1$ means one common value is sufficient; $\kdiag(G)=2$ means two weight classes suffice; and so forth. By Theorem~\ref{thm:existence}, $\kdiag(G)$ is finite for every weakly connected digraph. In particular, the convention $\kdiag(G)=0$ covers the one-vertex weakly connected digraph with no arcs.

\begin{remark}
The quantity in \eqref{eq:kdiag} measures weight \emph{heterogeneity}, not the number of arcs whose weights are independently selected. If $m$ arcs are allowed to leave a common-weight class, then at most $m+1$ values are needed, but several of those exceptional arcs may still share a value.
\end{remark}

\section{The uniform-weight benchmark: one common weight}

Let $L_0$ be the indegree Laplacian obtained by assigning unit weight to every arc. If every arc is instead given the same value $c>0$, then
\begin{equation}
L(c)=cL_0.
\end{equation}
Scalar multiplication by a nonzero number does not change diagonalizability.

\begin{proposition}
Let $G=(V,E)$ have $E\ne\varnothing$. If the uniform indegree Laplacian $L_0$ is diagonalizable, then
\begin{equation}
\kdiag(G)=1.
\end{equation}
\end{proposition}

\begin{proof}
If $S^{-1}L_0S=\Lambda$ is diagonal, then
\[
S^{-1}(cL_0)S=c\Lambda
\]
is diagonal for every $c>0$. Hence one common positive weight suffices.
\end{proof}

Every normal uniform Laplacian is covered automatically because normal matrices are diagonalizable over $\C$. Normality is only sufficient, however. A nonnormal Laplacian may also be diagonalizable. Therefore, the natural first algorithmic test is whether the \emph{uniform Laplacian itself} is diagonalizable, not merely whether it is normal.

\section{Directed acyclic graphs and longest-path levels}
\label{sec:dag}

Suppose now that $G$ is a DAG. A topological ordering labels the vertices $v_1,\ldots,v_n$ so that
\[
v_i\to v_j\quad\Longrightarrow\quad i<j.
\]
With the indegree convention, $L(w)$ is then lower triangular. Its eigenvalues are precisely its diagonal entries, namely the weighted indegrees. Sources contribute zero eigenvalues.

For each vertex $v$, define the longest-path level
\begin{equation}
\ell(v)=\max\{\text{length of a directed path from a source to }v\}.
\label{eq:level}
\end{equation}
The height of the DAG is
\begin{equation}
h(G)=\max_{v\in V}\ell(v).
\label{eq:height}
\end{equation}

\begin{lemma}
If $u\to v$ is an arc, then
\begin{equation}
\ell(v)\ge \ell(u)+1.
\end{equation}
More generally, if there is a directed path from $u$ to $v$, then $\ell(v)>\ell(u)$.
\end{lemma}

\begin{proof}
Take a longest path from a source to $u$ and append the arc $u\to v$. The resulting path has length $\ell(u)+1$, so $\ell(v)\ge\ell(u)+1$. Iterating the same argument along a directed path gives the second claim.
\end{proof}

Hence two vertices at the same level cannot be joined by a directed path. This simple observation is the key to reusing the same weight value on many arcs without necessarily creating a defective repeated eigenvalue.

\section{A level-based diagonalizing construction}
\label{sec:levelconstruction}

Choose positive numbers
\[
a_1,a_2,\ldots,a_{h(G)}.
\]
For every arc $u\to v$, assign
\begin{equation}
w_{uv}=a_{\ell(v)}.
\label{eq:levelweight}
\end{equation}
Thus the weight depends only on the level of the head vertex. At most $h(G)$ different values are used.

If $v$ is a nonsource with $\ell(v)=j$, every incoming arc has weight $a_j$. Hence
\begin{equation}
d_w^-(v)=d^-(v)a_j,
\end{equation}
and, because the Laplacian is triangular,
\begin{equation}
\lambda_v=d^-(v)a_{\ell(v)}.
\label{eq:eiglevel}
\end{equation}

Let
\begin{equation}
\Delta^- = \max_{v\in V} d^-(v).
\end{equation}
Choose any $M>\Delta^-$ and set
\begin{equation}
a_j=M^j,\qquad j=1,\ldots,h(G).
\label{eq:geomweights}
\end{equation}

Consider two nonsource vertices $u$ and $v$ with $\ell(u)=i<j=\ell(v)$. Then
\[
\lambda_u=d^-(u)M^i\le \Delta^-M^i<M^{i+1}\le M^j\le d^-(v)M^j=\lambda_v.
\]
Therefore vertices at different levels have different eigenvalues. Repetitions can occur only within one level.

The remaining issue is to show that such repetitions are harmless.

\begin{lemma}[Reachability-separated repetitions]
Let $L$ be the lower triangular indegree Laplacian of a weighted DAG. Fix an eigenvalue $\lambda$. If no directed path joins two distinct vertices whose diagonal entry equals $\lambda$, then $\lambda$ is semisimple.
\label{lem:semisimple}
\end{lemma}

\begin{proof}
Consider $(L-\lambda I)x=0$ in topological order. If the diagonal entry in the current row is not $\lambda$, the corresponding variable is uniquely determined by earlier variables. If the diagonal entry equals $\lambda$, the coefficient of the current variable vanishes and that row can impose only a constraint involving earlier variables.

A free variable introduced at one $\lambda$-vertex can influence a later row only through a sequence of nonzero off-diagonal entries. Such a sequence corresponds to a directed path in the DAG. By assumption, no $\lambda$-vertex can reach another $\lambda$-vertex. Hence a free variable created at one occurrence of $\lambda$ cannot be constrained at a later occurrence of the same eigenvalue.

Each diagonal occurrence of $\lambda$ therefore contributes one independent free variable. Because $L$ is triangular, the number of diagonal occurrences equals the algebraic multiplicity. Thus geometric and algebraic multiplicities coincide, and $\lambda$ is semisimple.
\end{proof}

\begin{theorem}[DAG height bound]\label{thm:dagheight}
Let $G$ be a DAG with height $h(G)$. Then
\begin{equation}
\boxed{\kdiag(G)\le h(G).}
\label{eq:dagbound}
\end{equation}
\end{theorem}

\begin{proof}
Choose $M>\Delta^-$ and assign weights according to \eqref{eq:levelweight}--\eqref{eq:geomweights}. Only the $h(G)$ values $M,M^2,\ldots,M^{h(G)}$ appear. Equation~\eqref{eq:eiglevel} and the choice $M>\Delta^-$ ensure that vertices at different levels have distinct eigenvalues.

If an eigenvalue is repeated, all vertices carrying that value lie at the same level. No directed path joins two vertices at the same level. Lemma~\ref{lem:semisimple} therefore makes every repeated eigenvalue semisimple. Hence the Laplacian is diagonalizable, proving \eqref{eq:dagbound}.
\end{proof}

The theorem can be substantially sharper than the crude strategy of assigning a different value to each non-source location. Whenever $h(G)\ll n-1$, many arcs can share weight values.

\section{Benchmark examples}
\label{sec:examples}

\subsection{Outward directed star}
Consider the outward star
\[
r\to v_1,\quad r\to v_2,\quad\ldots,\quad r\to v_{n-1}.
\]
Every leaf is at level one, so $h(G)=1$. A common positive weight gives the spectrum
\[
0,\underbrace{c,\ldots,c}_{n-1\ \text{times}},
\]
and the repeated eigenvalue is semisimple because no leaf reaches another leaf. Thus
\begin{equation}
\boxed{\kdiag(G)=1.}
\end{equation}
The star shows why demanding a simple spectrum would be unnecessarily strong for the present objective.

\subsection{Directed path}
Consider
\[
v_1\to v_2\to\cdots\to v_n,
\]
and let the arc $v_i\to v_{i+1}$ have weight $w_i>0$. The indegree Laplacian is
\begin{equation}
L=\begin{pmatrix}
0 & 0 & 0 & \cdots & 0\\
-w_1 & w_1 & 0 & \cdots & 0\\
0 & -w_2 & w_2 & \cdots & 0\\
\vdots & \vdots & \ddots & \ddots & \vdots\\
0 & 0 & \cdots & -w_{n-1} & w_{n-1}
\end{pmatrix},
\label{eq:pathlap}
\end{equation}
with eigenvalues $0,w_1,\ldots,w_{n-1}$.

Suppose $w_i=w_j=\lambda$ for some $i<j$. The head of the $i$th arc lies upstream of the head of the $j$th arc, and the unique path between them propagates the free variable created at the first $\lambda$-row to the second. At the second occurrence, the vanishing diagonal coefficient imposes a nontrivial constraint on that propagated variable. Hence the two occurrences do not produce two independent eigenvectors. The eigenvalue $\lambda$ is defective.

Therefore all path weights must be pairwise distinct, and
\begin{equation}
\boxed{\kdiag(\overrightarrow{P}_n)=n-1.}
\label{eq:path_exact}
\end{equation}
This is the extremal behaviour that the heterogeneity parameter is intended to detect.

\section{Exact formula for rooted outward directed trees}
\label{sec:trees}

Let $T$ be a rooted tree with every arc directed away from the root. Every nonroot vertex $v$ has exactly one incoming arc; let $w_v$ denote its weight. The corresponding nonzero Laplacian eigenvalue is simply
\[
\lambda_v=w_v.
\]
Thus repeated weight values are exactly repeated nonzero eigenvalues.

If two vertices carrying the same weight lie on one root-to-leaf chain, the same propagation argument used for the directed path shows that the repeated eigenvalue is defective. Conversely, if equal-weight vertices lie on different branches so that no directed path joins them, Lemma~\ref{lem:semisimple} makes the repeated eigenvalue semisimple. Hence equal values may be reused across branches, but not along one directed root-to-leaf chain.

\begin{theorem}[Exact outward-tree formula]
\label{thm:tree_exact}
For every rooted outward directed tree $T$,
\begin{equation}
\boxed{\kdiag(T)=h(T),}
\label{eq:tree_exact}
\end{equation}
where $h(T)$ is the maximum number of arcs in a root-to-leaf path.
\end{theorem}

\begin{proof}
The upper bound follows from Theorem~\ref{thm:dagheight}: assign one positive weight value to each positive level.

For the lower bound, take a longest root-to-leaf path. It contains $h(T)$ arcs. No two of these arcs may share a weight, because their head vertices are comparable by reachability and equal weights would create a defective repeated eigenvalue. Hence at least $h(T)$ distinct values are necessary. Combining both bounds gives \eqref{eq:tree_exact}.
\end{proof}

The two benchmark extremes follow immediately:
\[
\kdiag(\text{outward star})=1,
\qquad
\kdiag(\overrightarrow{P}_n)=n-1.
\]

\begin{example}[Balanced outward binary tree]
Consider a rooted outward binary tree with the root at level zero, two vertices at level one, and four vertices at level two. It has seven vertices and six arcs, but $h(T)=2$. Assign one value $a$ to the two arcs entering level one and a different value $b$ to the four arcs entering level two. Theorem~\ref{thm:tree_exact} gives
\[
\kdiag(T)=2.
\]
Thus six independently adjustable arc weights are unnecessary; two distinct values suffice.
\end{example}

\section{Relation to exceptional-arc counts}
\label{sec:tau}

A different parameter counts how many arcs must differ from an optimally chosen common value.

\begin{definition}[Minimum diagonalizing exceptional-arc count]
Let $G=(V,E)$ be a digraph. If $E=\varnothing$, define $\tdiag(G)=0$. Otherwise, provided that a positive weighting with diagonalizable Laplacian exists, define
\begin{equation}
\tdiag(G)=
\min_{\substack{c>0,\;w>0\\L(w)\ \mathrm{diagonalizable}}}
\left|\{e\in E:w_e\ne c\}\right|.
\label{eq:tdiag}
\end{equation}
\end{definition}

Thus $\tdiag(G)$ is the smallest number of exceptional arcs after both the common value $c$ and the diagonalizing positive weighting $w$ have been chosen. It is not the same as $\kdiag(G)$.

If $\tdiag(G)=m$, the nonexceptional arcs use the common value $c$, while the $m$ exceptional arcs use at most $m$ additional values. Therefore
\begin{equation}
\kdiag(G)\le m+1.
\label{eq:tau_kappa}
\end{equation}
The inequality can be strict because different exceptional arcs may share one value.

For the directed path on $n\ge2$ vertices, diagonalizability requires its $n-1$ arc weights to be pairwise distinct. For any choice of $c$, at most one arc can then have weight $c$; choosing $c$ equal to one of the arc weights attains this maximum. Hence
\[
\tdiag(\overrightarrow P_n)=n-2,
\]
whereas \eqref{eq:path_exact} gives
\[
\kdiag(\overrightarrow P_n)=n-1.
\]
There is no contradiction: one parameter counts exceptional arcs and the other counts distinct values.

The latter is the natural quantity when implementation complexity is tied to the number of different coupling gains rather than to the number of arcs that are independently tunable.

\section{An algorithmic weight-class viewpoint}
\label{sec:algorithmic}

The new problem can be viewed as a weight-class refinement problem. Start with one common class containing all arcs and split classes only when necessary.

\begin{enumerate}[label=\textbf{Step \arabic*.},leftmargin=*]
\item Form the uniform indegree Laplacian $L_0$. If $L_0$ is diagonalizable, stop: $\kdiag(G)=1$.

\item If $G$ is a DAG and uniform weighting is insufficient, compute the longest-path levels $\ell(v)$ and height $h(G)$. The construction of Section~\ref{sec:levelconstruction} gives a certified solution using at most $h(G)$ weight classes.

\item Attempt to merge level classes. The DAG height bound is guaranteed but need not be minimal. In a general DAG, a vertex may have several incoming arcs, so one small set of weight values can generate many different weighted indegrees through sums such as $a$, $2a$, $a+b$, and $2b$.

\item For a general digraph containing directed cycles, begin from the uniform weighting and identify defective eigenvalues or Jordan blocks. Split only selected arcs away from the common class, perturb them enough to destroy the defective structure, and then attempt to merge weight classes again.
\end{enumerate}

This leads to the exact combinatorial-spectral optimization problem
\begin{equation}
\boxed{
\text{partition }E\text{ into as few positive-weight classes as possible while keeping }L(w)\text{ diagonalizable.}
}
\end{equation}

The present paper solves this problem exactly for rooted outward trees and gives a general upper bound for DAGs.

\section{Reduced heterogeneity in directed-network master-stability analysis}
\label{sec:msf}

The role of $\kdiag(G)$ in master-stability analysis is deliberately limited and should not be confused with an extension of the MSF formalism itself.  The MSF was introduced by Pecora and Carroll \citep{PecoraCarroll1998}, and Nishikawa and Motter subsequently extended the formalism to nondiagonalizable network Laplacians through Jordan-form arguments \citep{NishikawaMotter2006}.  Directed networks therefore do not need to be diagonalizable in order to be treated by MSF methods.

What diagonalizability changes is the form of the variational equations.  When $L(w)$ is diagonalizable, the network perturbation splits into independent Laplacian modes.  With a nontrivial Jordan block, the corresponding generalized modes are coupled along a Jordan chain.  Choosing nonuniform positive arc weights to recover diagonalizability is therefore a design choice for obtaining the simpler, completely decoupled modal representation, not a prerequisite for MSF analysis.

Within that design choice, $\kdiag(G)$ has a direct implementation interpretation: it is the minimum number of distinct relative coupling strengths needed solely to obtain a diagonalizable weighted Laplacian.  The outward-star result shows that a highly branched outward tree needs only one value, while the directed path requires a different value on every arc.  More generally, the exact outward-tree formula links the number of required gain values to the longest root-to-leaf path length.

Diagonalizability by itself is not a synchronization certificate.  The scaled nonzero eigenvalues must still lie in a stable region of the MSF, and multiple root SCCs leave additional zero modes that positive reweighting cannot remove.  The heterogeneity problem studied here is therefore complementary to synchronization-oriented weight optimization \citep{SkardalTaylorSun2016} and to structural-genericity and perturbation results for weighted Laplacian spectra \citep{Poignard2018,PadePoignardPereira2019}.

\section{What is proved and what remains open}
\label{sec:open}

The results established here are:
\begin{enumerate}[label=(\roman*),leftmargin=*]
\item if the uniform indegree Laplacian is diagonalizable, then $\kdiag(G)=1$;
\item in particular, a normal uniform Laplacian satisfies $\kdiag(G)=1$;
\item every DAG satisfies $\kdiag(G)\le h(G)$;
\item an outward directed star satisfies $\kdiag(G)=1$;
\item the directed path satisfies $\kdiag(\overrightarrow P_n)=n-1$;
\item every rooted outward directed tree satisfies $\kdiag(T)=h(T)$.
\end{enumerate}

The natural next problem is not to conjecture $\kdiag(G)=h(G)$ for every DAG. A general DAG can have multiple incoming arcs at one vertex, so a small number of values may generate many distinct weighted indegrees. The useful question is instead
\begin{equation}
\boxed{
\text{How can the guaranteed DAG bound }\kdiag(G)\le h(G)\text{ be systematically reduced?}
}
\end{equation}
A satisfactory method should automatically detect $\kdiag(G)=1$ when the uniform Laplacian is already diagonalizable, while still recovering $\kdiag(\overrightarrow P_n)=n-1$ for the directed path.

Beyond DAGs, directed cycles introduce a second layer of difficulty because the spectrum is no longer read directly from the diagonal. It is therefore natural to study whether SCC decomposition, spanning-DAG reductions, and controlled restoration of cyclic arcs can provide useful upper bounds on $\kdiag(G)$ for arbitrary weakly connected digraphs. Another open direction is the complexity of computing $\kdiag(G)$ exactly.

\section{Conclusion}

Master-stability methods already apply to directed networks, including networks with nondiagonalizable Laplacians.  The question studied here is different: if positive nonuniform arc weights are deliberately used to recover a diagonalizable Laplacian and hence a fully decoupled modal representation, how much weight heterogeneity is actually necessary?

We formalized this question through $\kdiag(G)$, the minimum number of distinct positive arc-weight values needed to make the weighted indegree Laplacian diagonalizable.  The parameter separates the existence of a diagonalizing weighting from the number of different coupling values needed to implement one.  Uniformly diagonalizable topologies have $\kdiag(G)=1$ automatically.

For DAGs, a longest-path level construction gives the bound $\kdiag(G)\le h(G)$.  Its key feature is that repeated eigenvalues are allowed: weight values may be reused among same-level vertices because those vertices are not mutually reachable, making the repetitions semisimple.  For rooted outward directed trees the bound is exact.  Thus an outward star needs one weight value, whereas a directed path on $n$ vertices needs exactly $n-1$ distinct values.

The resulting viewpoint turns positive-weight diagonalization into a weight-class minimization problem.  It complements existing genericity results, which establish the abundance of favorable spectra after weight perturbation, by asking how economically a diagonalizing weighting can be represented.  The next steps are to reduce the DAG height bound for broader acyclic topologies and to develop comparable weight-class bounds for weakly connected digraphs containing directed cycles.

\printcredits

\section*{Funding}
This research did not receive any specific grant from funding agencies in the public, commercial, or not-for-profit sectors.

\section*{Data and code availability}
No external data were used in this study.

\section*{Declaration of competing interest}
The authors declare that they have no known competing financial interests or personal relationships that could have appeared to influence the work reported in this paper.

\section*{Declaration of generative AI and AI-assisted technologies in the manuscript preparation process}
During the preparation of this work, the authors used OpenAI ChatGPT and other generative-AI tools for preliminary proof exploration, language paraphrasing, LaTeX assistance, and suggestions concerning manuscript organization and presentation. All mathematical statements and proofs included in the manuscript were subsequently checked, corrected where necessary, and independently verified by the authors, who take full responsibility for the content of the article.

\bibliographystyle{unsrtnat}
\bibliography{references}

\end{document}